%% file: main.tex
\documentclass[11pt,letterpaper]{article}

\input{probheader}

\begin{document}
\maketitle

\begin{abstract}

We study the single-pass streaming complexity of deciding satisfiability of Constraint Satisfaction Problems (\textsf{CSPs}). A \textsf{CSP} is specified by a constraint language $\Gamma$, that is, a finite set of $k$-ary relations over the domain $[q] = \{0, \dots, q-1\}$. An instance of $\mathsf{CSP}(\Gamma)$ consists of $m$ constraints over $n$ variables $x_1, \ldots, x_n$ taking values in $[q]$. Each constraint $C_i$ is of the form $\{R_i,(x_{i_1} + \lambda_{i_1}, \ldots, x_{i_k} + \lambda_{i_k})\}$, where $R_i \in \Gamma$ and $\lambda_{i_1}, \ldots, \lambda_{i_k} \in [q]$ are constants; it is satisfied if and only if $(x_{i_1} + \lambda_{i_1}, \ldots, x_{i_k} + \lambda_{i_k}) \in R_i$, where addition is modulo $q$. In the streaming model, constraints arrive one by one, and the goal is to determine, using minimum memory, whether there exists an assignment 
satisfying all constraints.

For $k$-\textsf{SAT}, Vu (TCS 2024) proves an optimal $\Omega(n^k)$ space lower bound, while for general \textsf{CSPs}, Chou, Golovnev, Sudan, and Velusamy (JACM 2024) establish an $\Omega(n)$ lower bound; a complete characterization has remained open. We close this gap by showing that the single-pass streaming space complexity of $\mathsf{CSP}(\Gamma)$ is precisely governed by its \emph{non-redundancy}, a structural parameter introduced by Bessiere, Carbonnel, and Katsirelos (AAAI 2020). The non-redundancy $\mathsf{NRD}_n(\Gamma)$ is the maximum number of constraints over $n$ variables such that every constraint is \emph{non-redundant}, i.e., omitting it strictly expands the set of satisfying assignments.
We prove that the single-pass streaming complexity of $\mathsf{CSP}(\Gamma)$ is characterized, up to a logarithmic factor, by $\mathsf{NRD}_n(\Gamma)$. We also extend this characterization to \emph{positive} Boolean \textsf{CSPs}, i.e., instances in which no additive shifts are applied, a class that includes graph $2$-colorability (equivalently, bipartiteness) as a canonical example.

A key ingredient in our lower bound proof is a binary relation $\mathcal{E}$ on the set of assignments $[q]^n$, where $(a, b) \in \mathcal{E}$ if every instance satisfied by $a$ is also satisfied by $b$. While it is immediate that $\mathcal{E}$ is reflexive and transitive, its symmetry, which would make it an equivalence relation, is non-trivial. We show that $\mathcal{E}$ is an equivalence relation for general \textsf{CSPs}, and for positive Boolean \textsf{CSPs} when excluding the two \emph{constant assignments} ($0^n$ and $1^n$). We believe this equivalence structure could be of independent interest.

\end{abstract}

\newpage
\tableofcontents

\section{Introduction}

Constraint satisfaction problems (\CSPs) provide a unifying framework that captures a wide range of computational problems, and serve as a central object of study in theoretical computer science. They have been studied extensively from the viewpoints of complexity theory 
\cite{Schaefer78,jeavons1997closure,feder1998computational,Hastad01,Khot02,bulatov2005classifying,Chen06SchaeferExposition,zhuk2020proof}, 
approximation \cite{khanna1997complete,khanna2001approximability,raghavendra2008optimal,austrin2009approximation,DBLP:conf/stoc/KhotTW14,CLRS16,GT17,DBLP:conf/focs/AlevJT19},
(approximate) sparsification 
\cite{kogan2015sketching,filtser2017sparsification,butti2020sparsification,KPS24,khanna2025efficient,khanna2025theory}
,
kernelization and exact sparsification \cite{dell2014satisfiability,jansen2019optimal,chen2020best,lagerkvist2020sparsification, bessiere2020chain,carbonnel2022redundancy,brakensiek2025redundancy,brakensiek2025richness,brakensiek2025tight}, and, more recently, streaming and sketching \cite{KKS15,GVV17,KK19,CGV20,CGS+22-linear-space,Sin23-kand,KolParamonovSaxenaYu-ITCS23,Vu-2208.09160,CGSV24,STV25,FMW26,FMW26-tight,ABF26,STV26}. Among the most basic objectives for a \CSP\ instance are deciding \emph{satisfiability} 
\cite{Schaefer78,feder1994network,jeavons1997closure,feder1998computational,Hastad01,Khot02,bulatov2005classifying,Chen06SchaeferExposition,zhuk2020proof,CGSV24,Vu-2208.09160}, 
maximizing the number of satisfied constraints (\emph{\maxcsp})
\cite{khanna1997complete,khanna2001approximability,Hastad01,lewin2002improved,Khot02,KKMO07,raghavendra2008optimal,austrin2009approximation,manurangsi2015approximating,fotakis2016sub,jeronimo2020unique,bafna2021playing,jeronimo2021near,KolParamonovSaxenaYu-ITCS23,CGSV24},
and minimizing the number of unsatisfied constraints\footnote{Although \maxcsp and \mincsp coincide in the exact optimization setting, their approximation behaviours are markedly different, as we discuss later.} (\emph{\mincsp})
\cite{khanna2001approximability,bazgan2003polynomial,ACMM05,cohen2006complexity,KS09,guruswami2011lasserre,Vu-2208.09160,als25,alms25}.
In the streaming literature, while \maxcsps have been extensively studied, with several sharp dichotomy characterizations \cite{CGV20,KolParamonovSaxenaYu-ITCS23,CGSV24,FMW26,STV26,ABF26,FMW26-tight}, 
the satisfiability and \mincsp variants remain relatively under-explored. In this work, we study the streaming satisfiability problem for finite-domain \CSPs.

A constraint satisfaction problem is defined over a finite domain $[q] := \{0,1,\dots,q-1\}$ and specified by a finite constraint language $\Gamma$ consisting of $k$-ary relations over $[q]$\footnote{We think of $k$ and $q$ as constants.}. An instance of $\CSP(\Gamma)$ consists of $n$ variables $x_1,\dots,x_n$ taking values in $[q]$ and $m$ constraints $C_1,\dots,C_m$, where each $C_i$ has the form $\{R_i,(x_{i_1}+\lambda_{i_1},\dots,x_{i_k}+\lambda_{i_k})\}$ for some relation $R_i \in \Gamma$, $k$-tuple of variables $(x_{i_1},\dots,x_{i_k})$, and constants $\lambda_{i_1},\dots,\lambda_{i_k} \in [q]$. An assignment $a \in [q]^n$ satisfies $C_i$ if $(a_{i_1}+\lambda_{i_1},\dots,a_{i_k}+\lambda_{i_k}) \in R_i$ (with all arithmetic modulo $q$), and satisfies the instance if it satisfies every constraint.
We say that an instance is satisfiable if there exists an assignment that satisfies it.
In the streaming model, constraints arrive one by one, and the 
goal is to decide if the instance is satisfiable, using as little memory as 
possible. 
A deterministic streaming algorithm must succeed on every input, while a randomized 
algorithm is required to succeed with probability at least $2/3$ on every input.
Since any instance has at most $O(n^k)$ constraints, 
storing them all gives a trivial $\widetilde{O}(n^k)$ upper bound. For \ksat, Vu~\cite{Vu-2208.09160} 
proved a matching $\Omega(n^k)$ lower bound, showing this is essentially tight. For 
more general \csps, however, the best known lower bound is $\Omega(n)$, due to Chou, 
Golovnev, Sudan, and Velusamy~\cite{CGSV24}. This gap motivates our central question:

\begin{center}
\emph{Can we completely characterize the single-pass streaming complexity of deciding 
satisfiability of \CSPs?}
\end{center}

We answer this in the affirmative, obtaining a complete characterization in terms of 
a structural parameter called \emph{non-redundancy}. For $\CSP(\Gamma)$, the 
non-redundancy $\nrdg$ is the size of the largest instance over $n$ variables 
in which every constraint is \emph{non-redundant}, i.e., removing any single constraint 
strictly increases the number of satisfying assignments. Equivalently, every constraint 
admits a \emph{witness} assignment that violates it but satisfies all other constraints.

As a simple example, $\nrdx{\twosat} = \Omega(n^2)$. Consider a complete graph on 
$n$ vertices, with each edge $(i,j)$ defining the constraint $x_i \vee x_j$.\footnote{In our notation, \twosat~is specified by the single binary relation $R = \{(0,1),(1,0),(1,1)\}$, and an OR constraint such as $x_i \vee \neg x_j$ is expressed as $\{R, (x_i, x_j + 1)\}$, where the shift of $1$ on $x_j$ encodes the negation $\neg x_j$.} The 
assignment $x_i = x_j = 0$ and $x_k = 1$ for all $k \neq i,j$ satisfies every other 
constraint but violates this one. The same construction also generalizes naturally to 
$k$-\textsf{SAT}. Thus, Vu's result is a special case of our main theorem.

\begin{restatable}[Main result]{theorem}{mainresult}\label{thm:main}
For every constraint language $\Gamma$ containing $k$-ary relations  over $[q]$, to solve the satisfiability of \cspg, there exists:
\begin{enumerate}[label=(\alph*)]
    \item\label{main-a} a deterministic  streaming algorithm in $O_{q,k}(\nrdg \log n)$ space, and
    \item\label{main-b} no randomized streaming algorithm in $o(\nrdg)$ space.
\end{enumerate}
\end{restatable}

Non-redundancy also plays a central role in the sparsification theory of \csps. The 
concept was first introduced by Bessiere, Carbonnel, and Katsirelos \cite{bessiere2020chain} in the 
context of learning \csp instances from membership queries, and has since emerged as a 
fundamental structural parameter underlying exact sparsification and kernelization 
\cite{dell2014satisfiability,jansen2019optimal,chen2020best,lagerkvist2020sparsification,carbonnel2022redundancy}. 
Very recently, it was also shown to govern approximate sparsification up to 
polylogarithmic factors \cite{brakensiek2025redundancy,brakensiek2025richness,brakensiek2025tight}. 
Our work adds to this growing line of research by showing that non-redundancy also 
characterizes the streaming complexity of deciding \csp satisfiability.

We also extend our characterization to \emph{positive} Boolean \CSPs, where constraints 
are applied only to variables and not their negations. This well-studied special class 
captures problems such as $2$-colorability that fall outside the previous definition. 
While the upper bound follows immediately from \Cref{thm:main}, the lower bound requires 
a different argument, as constraints can no longer be applied to negated variables.
An additional subtlety arises in this setting: without literals, satisfiability can 
become trivial for some constraint languages. For example, consider positive \twolin, 
which contains the single binary relation $R = \{(1,1),(0,0)\}$. Although 
$\NRD_n(\twolin) = \Theta(n)$, every positive instance of \twolin~is trivially 
satisfied by the constant assignments $0^n$ and $1^n$. Excluding \emph{constant-satisfiable} 
languages, where every instance is satisfied by some constant assignment, we still 
obtain a tight characterization based on non-redundancy. We formally state our result 
below.

\begin{theorem}
For every non constant-satisfiable Boolean constraint language $\Gamma$ containing $k$-ary Boolean relations, no randomized streaming 
algorithm can decide satisfiability of positive $\CSP(\Gamma)$ in 
$o(\nrdg)$ space.
\end{theorem}

We now discuss how the different \csp objectives relate to one another and how our 
results compare to those obtained under other objectives.

We first consider the \maxcsp objective, where the goal is to find an assignment 
maximizing the number of satisfied constraints. In the polynomial time setting, 
satisfiability can be significantly easier than \maxcsp. For example, \maxx-\twosat~is \textsf{NP}-hard while satisfiability of a \twosat~instance is decidable in polynomial 
time. Our results show for the first time that an analogous phenomenon occurs in 
the streaming setting, though interestingly not for \twosat.\footnote{Surprisingly, 
the streaming complexities of decidability of \twosat~and \maxx-\twosat~coincide.} Specifically, Kol, Paramonov, Saxena, and 
Yu \cite{KolParamonovSaxenaYu-ITCS23} give a complete characterization of the streaming complexity of \maxcsps, showing that for \maxx-$k$-\lin, the \csp where constraints correspond to 
linear equations over $\mathbb{F}_2$, the space complexity is $\Theta(n^k)$, up to logarithmic factors. In contrast, we show that 
satisfiability of $k$-\lin can be decided deterministically in $\widetilde{O}(n)$ space.

Next, we consider the \minn-\csp objective, where the goal is to find an assignment 
minimizing the number of unsatisfied constraints. Although \mincsp and \maxcsp are 
equivalent in the \emph{exact} setting, since minimizing unsatisfied constraints is 
the same as maximizing satisfied ones, they differ significantly in the 
\emph{approximation} setting.
For $\alpha \in (0,1)$, a deterministic $\alpha$-approximation algorithm for \maxcsp 
outputs, on every instance, an assignment satisfying at least an $\alpha$-fraction of 
the maximum number of satisfiable constraints. A randomized approximation algorithm is 
defined analogously, except it is only required to succeed with probability at least 
$2/3$ on every input. For \mincsp, an $\alpha$-approximation algorithm outputs an 
assignment that unsatisfies at most a $1/\alpha$-fraction of the minimum number of 
unsatisfied constraints. Observe that if $\alpha$ is non-zero and the input instance 
is satisfiable, then the assignment output by the $\alpha$-approximation algorithm for 
\mincsp must itself be satisfying. Thus, \mincsp is at least as hard as deciding 
satisfiability, and all our lower bounds immediately carry over to any non-trivial 
approximation of \mincsp.

While a standard edge-sampling algorithm achieves an arbitrary constant-factor 
approximation for \maxcsp in linear space \cite{AG09}, we show for several \mincsp 
problems that any non-trivial approximation requires \emph{super-linear} space. 
For instance, graph $q$-coloring (equivalently, the relation $x \neq y$ over $[q]$ 
for $q \geq 3$), as well as other graph and digraph homomorphism problems, satisfy 
$\NRD_n(\Gamma) = \Theta(n^2)$ \cite{bessiere2020chain,carbonnel2022redundancy}, and 
thus admit quadratic space lower bounds.

Finally, we observe that even for approximating \minn-$k$-\lin, our lower bound is 
essentially tight. Specifically, a streaming algorithm for code sparsification due to 
Kuchlous~\cite{kuchlous2025} can be adapted to obtain $(1-\epsilon)$-approximation for \minn-$k$-\lin in near-linear space for every constant $\epsilon > 0$.
Whether non-redundancy also governs the streaming approximability of general \mincsps remains an interesting open problem.

\subsection{Technical overview}\label{sec:tech-overview}
We now give a brief technical overview of the proofs of our upper and lower bounds.

\paragraph*{Upper bound.} We give a simple deterministic greedy algorithm that maintains a non-redundant sub-instance with the same set of satisfying assignments as the original instance. Upon the arrival of each constraint, the algorithm checks 
whether it is redundant with respect to the stored sub-instance, and retains it if not. Since the number of stored constraints cannot exceed the non-redundancy of the \csp, the upper bound follows immediately.

\paragraph*{Lower bound.} The lower bound proof is technically more involved. The 
standard approach for proving streaming lower bounds is via a reduction from a hard 
one-way communication problem. 
A space-efficient streaming algorithm can be converted 
into an efficient one-way communication protocol with communication complexity 
matching the space complexity of the streaming algorithm. We obtain our reduction from 
one of the most well-studied problems in communication complexity, the \Index problem. 
Here, Alice receives an $N$-bit string $S$ and Bob receives an index $i \in [N]$; 
Alice sends a single message to Bob, who must then output the bit $S_i$. It is 
well-known that any randomized one-way protocol for \Index that succeeds with 
probability at least $2/3$ requires $\Omega(N)$ bits of 
communication~\cite{Index1,Index2}.

As a starting point, we describe our reduction for \twosat. We consider the \Index 
problem with $N = \binom{n}{2}$ and interpret Alice's input as a graph $G$ on $n$ 
vertices and Bob's input as an edge query $(i,j)$. Together, they must determine 
whether $(i,j)$ is an edge in $G$. We reduce this to the streaming decidability of 
\twosat~as follows. Alice constructs an instance $I_G$ of \twosat\ on $n$ variables by adding the 
constraint $x_u \vee x_v$ for every edge $(u,v)$ in $G$. Bob constructs a 
two-constraint instance $I_{(i,j)}$ consisting of $\neg x_i \vee \neg x_i$ and 
$\neg x_j \vee \neg x_j$, which is satisfied if and only if $x_i = x_j = 0$. Thus, 
$I_G \cup I_{(i,j)}$\footnote{Here, we are taking the union of all the constraints.} is satisfiable if and only if $(i,j)$ is not an edge in $G$, 
and therefore any randomized streaming algorithm deciding satisfiability of \twosat~must use $\Omega(n^2)$ space.
We note that Vu~\cite{Vu-2208.09160} also proves an optimal lower bound for \ksat~via a reduction from the \Index problem, though his reduction is different and 
tailored\footnote{While the reduction in \cite{CGSV24} is more general than that of \cite{Vu-2208.09160}, it relies on a different communication problem, namely disjointness, and is incomparable to ours. Moreover, their approach yields only an $\Omega(n)$ lower bound.} more specifically to \ksat. The reduction described above, by contrast, 
generalizes naturally to arbitrary \csps, as we now describe.

Recall that a complete graph on $n$ vertices, with each edge $(i,j)$ defining the 
constraint $x_i \vee x_j$, is a non-redundant instance of \twosat; call this instance 
$I_{K_n}$. Alice's input can thus be viewed as a sub-instance $I_G$ of $I_{K_n}$, 
and Bob's input as a constraint $C$ in $I_{K_n}$. In the reduction, Bob constructs an 
instance $I_C$ such that $I_G \cup I_C$ is satisfiable if and only if $C \notin I_G$.
We show that an analogous reduction is feasible for general \csps. Specifically, we 
start with the largest non-redundant instance $I_n$ of $\CSP(\Gamma)$ on $n$ variables, where Alice 
receives a sub-instance $I$ of $I_n$ and Bob receives a constraint $C$ in $I_n$. We 
then demonstrate that Bob can always construct an instance $I_C$ such that $I \cup I_C$ 
is satisfiable if and only if $C \notin I$. Since $I_n$ is non-redundant, every constraint 
$C$ in $I_n$ admits a witness assignment $a_C \in [q]^n$ that unsatisfies $C$ but 
satisfies every other constraint in $I_n$. If Bob can construct an instance $I_C$ whose 
\emph{unique} satisfying assignment is $a_C$, then we are done. While this is possible 
for \twosat, it need not hold for other \csps. For example, for \twolin, if $a$ is a 
satisfying assignment then its complement $\overline{a}$ is also satisfying. This is 
not an obstacle, however, if $a$ is a witness assignment then so is $\overline{a}$, 
and the same reduction goes through.

For general CSPs, this motivates us to define a binary relation $\equivr$ on the set of assignments 
$[q]^n$, where $(a, b) \in \equivr$ if every instance satisfied by $a$ is also 
satisfied by $b$. While reflexivity and transitivity are immediate, symmetry, which 
would make $\equivr$ an equivalence relation, is not. We establish the following.

\begin{theorem}[informal, restatement of \Cref{thm:equiv-lit,thm:equiv-boolean-no-lit}]
\enspace
\begin{itemize}
    \item For general \CSPs, $\equivr$ is an equivalence relation on all 
    assignments;
    \item For positive Boolean \CSPs, $\equivr$ is an equivalence relation when 
    restricted to non-constant assignments.
\end{itemize}
\end{theorem}

With the above theorem, the reduction for general \csps follows immediately. Bob 
constructs an instance $I_C$, which is the union of all instances that have $a_C$ as a satisfying assignment. By 
the equivalence structure of $\equivr$, it is not hard to show that the satisfying 
assignments of $I_C$ are exactly the equivalence class of $a_C$, and the reduction 
follows. The argument is slightly more subtle for positive Boolean \csps, since the 
equivalence relation does not hold for all assignments. Nevertheless, this case is 
manageable. Observe that no two constraints in a non-redundant instance can share the same witness assignment, 
by definition. Therefore, we can remove the constraints whose witness assignment is 
$0^n$ or $1^n$, yielding a non-redundant instance of asymptotically the same size, 
for which the same reduction goes through.

Finally, we show in \cref{sec:counter-gen-nolit} that our reductions do not extend to positive \CSPs over larger 
domains, i.e., $q>2$. The failure is not merely a limitation of our proof technique via the relation 
$\equivr$; rather, the lower bound itself is false in general. We provide two 
counterexamples over non-Boolean domains: one using a non-singleton language 
(\Cref{sec:counter-anyarity}) and another using a singleton language 
(\Cref{sec:counter-singleton}), for which single-pass streaming satisfiability is 
strictly easier than $\nrdg$ and the relation $\equivr$ fails to be symmetric.

\subsection{Organization}
We begin in \Cref{sec:prelims} with notation and basic definitions. In \Cref{sec:results}, we state our main results and the general streaming upper bounds. The lower bound for general \CSPs is proved in \Cref{sec:lower-gen}, where we establish the equivalence-class structure of the relation $\equivr$ and derive the streaming lower bound from it. We then turn in \Cref{sec:lower-bool-nolit} to positive Boolean \CSPs, proving that $\equivr$ remains an equivalence relation on non-constant assignments and using this to obtain the corresponding lower bound. In \Cref{sec:counter-gen-nolit}, we show that this phenomenon does not extend to positive \CSPs on larger alphabets.

\section{Preliminaries and notation}\label{sec:prelims}

\subsection{Constraint languages and instances}

Fix a domain $[q] := \{0,1,\dots,q-1\}$ for some integer $q \ge 2$. A constraint language $\Gamma$ over $[q]$ is a set of relations over $[q]$. Let $k$ denote the maximum arity of a relation in $\Gamma$. By padding each relation $R \subseteq [q]^\ell$ with $\ell < k$ to $\widehat{R} := R \times [q]^{\,k-\ell} \subseteq [q]^k$, and using padded variables in constraint scopes, we may assume without loss of generality that every relation in $\Gamma$ has arity exactly $k$. Throughout the paper, we only consider finite constraint languages of bounded arity, i.e., languages for which the maximum arity $k$ is bounded by a constant.

An instance $I$ of $\CSP(\Gamma)$ on variables $x_1,\dots,x_n$ consists of a set of $m$ constraints. For each $i \in [m]$, the constraint $C_i$ is of the form $\{R_i,(x_{i_1}+\lambda_{i_1},\dots,x_{i_k}+\lambda_{i_k})\}$, where $R_i \in \Gamma$, $(i_1,\dots,i_k)\in [n]^k$, and $\lambda_{i_1},\dots,\lambda_{i_k} \in [q]$, with all arithmetic modulo $q$. An assignment $a \in [q]^n$ satisfies such a constraint if $(a_{i_1}+\lambda_{i_1},\dots,a_{i_k}+\lambda_{i_k})\in R_i$. We write $\Sat(I)\subseteq [q]^n$ for the set of satisfying assignments of $I$. Since duplicate constraints do not affect satisfiability, we identify an instance with the set of its constraints. Unless explicitly stated otherwise, all \CSP instances in this paper are stated in this general model.

We call a constraint language $\Gamma$ \emph{non-trivial} if every relation in $\Gamma$ is non-empty (i.e., not trivially unsatisfiable) and at least one relation in $\Gamma$ is proper (i.e., not trivially satisfiable\footnote{A relation is trivially satisfiable if it is satisfied by every assignment.}). For streaming satisfiability, this assumption is without loss of generality because constraints using trivially satisfiable relations can simply be ignored, while the appearance of a constraint using an empty relation makes the instance immediately unsatisfiable. Thus, we restrict our attention to non-trivial languages throughout the paper.



\paragraph{Positive \CSPs.}  We call a \CSP a \emph{positive \CSP} if every constraint $C_i$ is applied with the all the shifts $\lambda_{i_j}$ set to zero, so every constraint has the form $\{R,(x_{i_1},\dots,x_{i_k})\}$.
We say that for a positive \cspg, the underlying constraint language $\Gamma$ is \emph{constant-satisfiable} if there exists some $\alpha \in [q]$ such that $\alpha^k\in R$ for every relation $R \in \Gamma$. Equivalently, the all-$\alpha$ assignment trivially satisfies every instance of positive $\CSP(\Gamma)$. If no such $\alpha$ exists, then $\Gamma$ is \emph{non constant-satisfiable}. In our lower bound proof for positive \CSPs, we will assume that $\Gamma$ is non constant-satisfiable (note this is not a default assumption for general \csps).




\subsection{Equivalent instances, redundancy, and witnesses}
Two instances $I$ and $J$ of $\CSP(\Gamma)$ are \emph{equivalent} if $\Sat(I)=\Sat(J)$. A constraint $C\in I$ is \emph{redundant} if deleting it does not change the set of satisfying assignments, that is, $\Sat(I\setminus\{C\})=\Sat(I)$. Equivalently, $C$ is redundant if every assignment satisfying all constraints in $I\setminus\{C\}$ also satisfies $C$. We say that $I$ is non-redundant if none of its constraints is redundant.

The main combinatorial parameter in this paper is the maximum size of a non-redundant instance on $n$ variables. For a constraint language $\Gamma$, we define 

\[\nrdg:=\max\{|I| : I \text{ is a \emph{non-redundant instance of }}  \CSP(\Gamma)\text{ on } n \text{ variables}\}.\]

If $I=\{C_1,\dots,C_m\}$ is non-redundant, then for each $i\in[m]$ there exists an assignment that satisfies every constraint in $I\setminus\{C_i\}$ but unsatisfies $C_i$. We call any such assignment a \emph{witness assignment} for $C_i$. Note that any sub-instance of a non-redundant instance is also non-redundant.

Any non-trivial $\Gamma$ contains a relation $R$ with $\emptyset \neq R \subsetneq [q]^k$. 
We can partition the variables into $ \lfloor n/k \rfloor$ pairwise disjoint $k$-tuples, and apply the relation $R$ on each tuple. It is easy to see that each constraint in this set is non-redundant. 
Thus, we have $\NRD_n(\Gamma) \ge \lfloor n/k \rfloor\ge \Omega(n)$. 



\subsection{Streaming problems}
\paragraph{Streaming Satisfiability.} Fix a constraint language $\Gamma$. In the single-pass \emph{streaming satisfiability} problem, the input is a stream of constraints of an instance of $\CSP(\Gamma)$ on $n$ variables, and at the end of the stream, the algorithm must decide whether the instance is satisfiable. A deterministic streaming algorithm must succeed on every input, while a randomized algorithm is required to succeed with probability at least $2/3$ on every input.

\paragraph{Streaming Enumeration and Exact Sparsification.} We also consider the single-pass \emph{streaming enumeration} problem, in which the algorithm must output the full satisfying set $\Sat(I)$, or output $\emptyset$ if the instance is unsatisfiable. We also consider the single-pass \emph{streaming exact sparsification} problem, in which the algorithm must output an equivalent sub-instance $J \subseteq I$ of the streamed instance $I$, that is, a sub-instance satisfying $\Sat(J)=\Sat(I)$. For these problems, we consider a separate write-only output tape that does not count toward the space bound.

\paragraph{General notation.}
In the Boolean setting, the alphabet is $\{0,1\}=\{0,1\}$. For $n\ge 1$, we write $0^n$ and $1^n$ for the two constant assignments in $\{0,1\}^n$. An assignment $a\in\{0,1\}^n$ is \emph{non-constant} if $a\notin\{0^n,1^n\}$. All asymptotic notations are with respect to the number of variables $n$. In particular, $\widetilde{O}(\cdot)$, $\widetilde{\Omega}(\cdot)$, and $\widetilde{\Theta}(\cdot)$ suppress factors polylogarithmic in $n$, and the hidden constants may depend on the fixed alphabet size and constraint language.

\section{Main Results and Upper Bounds}\label{sec:results}

We now state our main results. The first theorem gives the complete characterization for general $\CSP$s, up to a logarithmic factor, the single-pass streaming space complexity of deciding satisfiability is exactly governed by the non-redundancy parameter.

\mainresult*


We next turn to positive Boolean $\CSP$s. In this setting, for non constant-satisfiable languages, the same bound, tight up to a logarithmic factor, continues to hold.

\begin{restatable}[Main positive Boolean \csp  result]{theorem}{mainbool}
\label{thm:main-bool-nolit}
For every non constant-satisfiable Boolean constraint language $\Gamma$ containing $k$-ary Boolean relations, to solve the satisfiability of positive \cspg there exists:
\begin{enumerate}[label=(\alph*)]
    \item\label{mainbool-a} a deterministic streaming algorithm in $O_{q,k}(\nrdg \log n)$ space; and
    \item\label{mainbool-b} no randomized algorithm in $o(\nrdg)$ space.
\end{enumerate}
\end{restatable}


We first prove the upper bounds. Throughout the stream, our algorithm is to maintain an equivalent non-redundant sub-instance of the constraints seen so far. The same algorithm works in both settings.

\begin{proof}[Proof of \Cref{thm:main} \ref{main-a} and \Cref{thm:main-bool-nolit} \ref{mainbool-a}]
Let $F_t$ denote the set of first $t$ constraints seen in the stream. After processing the first $t$ constraints, the algorithm stores a non-redundant instance $K_t \subseteq F_t$ such that $\Sat(K_t)=\Sat(F_t)$. Initially $K_0=\emptyset$.

When a new constraint $C$ arrives, consider the instance $J:=K_t \cup \{C\}$. Repeatedly delete from $J$ any constraint that is redundant with respect to the current instance, and let the resulting non-redundant instance be $K_{t+1}$. Since we only delete redundant constraints, we have
\[
\Sat(K_{t+1})=\Sat(J)=\Sat(K_t)\cap \Sat(C)=\Sat(F_t)\cap \Sat(C)=\Sat(F_{t+1}).
\]

By construction, each $K_t$ is non-redundant, so $|K_t| \le \nrdg$. Since each relation in $\Gamma$ has arity $k$, each stored constraint uses $O_{q,k}(\log n)$ bits, and therefore storing $K_t$ takes $O_{q,k}(\nrdg\log n)$ bits.

It remains to show that each update can be performed within the same space bound. To test whether a constraint $C' \in J$ is redundant, it suffices to check whether there exists an assignment satisfying all constraints in $J \setminus \{C'\}$ while falsifying $C'$. 
Since we impose no time bound, we may brute-force over all assignments to the $n$ variables using only a counter and the current assignment, and this requires $O_{q}(n)$ bits of space. Since, $\nrdg\geq\Omega(n)$, every update can be carried out in $O_{q,k}(\nrdg)$ space.
At the end of the stream, the algorithm stores an equivalent non-redundant sub-instance $K_T$ with $\Sat(K_T)=\Sat(I)$. Therefore $I$ is satisfiable if and only if $K_T$ is satisfiable. 
We can decide whether $\Sat(K_T)\neq \emptyset$ by brute force within $O_{q,k}(\nrdg\log n)$ space. This proves the theorem.
\end{proof}


This proves the deterministic upper bound for any general \CSP, and thus extends trivially to positive Boolean \csps as well. The lower bounds are proved in \Cref{sec:lower-gen} for general \CSPs and \Cref{sec:lower-bool-nolit} for positive Boolean \CSPs.



\begin{remark}
The algorithm in \Cref{thm:main} solves a more general problem than 
satisfiability: it maintains an equivalent non-redundant sub-instance throughout the 
stream. Consequently, it yields a deterministic single-pass streaming algorithm for 
exact sparsification and streaming enumeration in $O_{q,k}(\nrdg \log n)$ space.
Conversely, both exact sparsification and enumeration are at least as hard as 
satisfiability, since an equivalent sub-instance preserves satisfiability exactly, and 
the full set of satisfying assignments reveals whether the instance is satisfiable. 
Thus the same $\Omega(\nrdg)$ lower bound applies to these tasks as well.
\end{remark}

\section{Lower Bounds}\label{sec:lower}

In this section we prove the streaming lower bounds stated in \Cref{thm:main} \ref{main-b} and \Cref{thm:main-bool-nolit} \ref{mainbool-b}. Both arguments proceed by reduction from the standard one-way communication problem $\Index_m$ in which Alice receives a string $x \in \{0,1\}^m$, Bob receives an index $i \in [m]$, and Bob must output $x_i$ after receiving a single message from Alice. We use the following standard fact.

\begin{lemma}\cite{Index1,Index2}\label{lem:index-lb}
Every public-coin one-way communication protocol for $\Index_m$ with error probability 
at most $1/3$ requires $\Omega(m)$ bits of communication.
\end{lemma}

\subsection{General \texorpdfstring{\csps}{CSPs}}\label{sec:lower-gen}

Our lower bound for $\CSP(\Gamma)$ (\Cref{thm:main} \ref{main-b}) is based on a structural relation on assignments. For $a,b \in [q]^n$, we write $(a,b)\in \equivr$ if every instance of $\CSP(\Gamma)$ that is satisfied by $a$ is also satisfied by $b$. This relation is immediately reflexive and transitive, but symmetry is far from obvious. The key point of the next theorem is that, for every constraint language $\Gamma$, this relation is in fact an equivalence relation. 
Moreover, for every equivalence class, we show that there exists an instance whose set of satisfying assignments is exactly that class.




\begin{theorem}[Equivalence Class Theorem for $\CSP(\Gamma)$]\label{thm:equiv-lit}
Let $n \in \mathbb{N}$ and $\Gamma$ be a constraint language containing $k$-ary relations over the domain $[q]$. For $a,b \in [q]^n$, define $(a,b)\in \equivr$ if every instance of $\CSP(\Gamma)$ on $n$ variables that is satisfied by $a$ is also satisfied by $b$. Then $\equivr$ is an equivalence relation on $[q]^n$.

As a consequence, for every $a \in [q]^n$, there exists an instance $P_a$ of $\CSP(\Gamma)$ such that $\Sat(P_a) = [a]_{\equivr}$, where $[a]_{\equivr}$ denotes the equivalence class of $a$ under $\equivr$.
\end{theorem}
\begin{proof}
Reflexivity and transitivity are immediate, so it remains to prove symmetry. Fix $a,b \in [q]^n$ with $(a,b)\in \equivr$. For any $\Delta \in [q]^n$, define the shift operator $\tau_\Delta$ on constraints by
\[
\tau_\Delta\bigl(\{R,(x_{i_1}+\lambda_{i_1},\dots,x_{i_k}+\lambda_{i_k})\}\bigr)
:=
\{R,(x_{i_1}+\lambda_{i_1}+\Delta_{i_1},\dots,x_{i_k}+\lambda_{i_k}+\Delta_{i_k})\},
\]
with all additions modulo $q$, and extend $\tau_\Delta$ to instances by applying it each individual constraint in the instance.
Then for every assignment $z$ and every instance $J$, we have that $z$ satisfies $\tau_\Delta(J)$ if and only if $z+\Delta$ satisfies $J$.

Let $\mathcal{S}_a$ denote the set of all instances of $\CSP(\Gamma)$ on $n$ variables that are satisfied by $a$, and define $\mathcal{S}_b$ analogously. Now let $\Delta:=a-b$. Then for every instance $J$, we have that  $a \sats J$ if an only if $\quad b \sats \tau_\Delta(J)$, since $b+\Delta=b+(a-b)=a$. Thus $\tau_\Delta$ maps $\mathcal{S}_a$ bijectively onto $\mathcal{S}_b$, with inverse $\tau_{-\Delta}$, and therefore $|\mathcal{S}_a|=|\mathcal{S}_b|$.

On the other hand, since $(a,b)\in \equivr$, every instance satisfied by $a$ is also satisfied by $b$, so $\mathcal{S}_a \subseteq \mathcal{S}_b$. Because $\Gamma$ contains $k$-ary relation, there are only finitely many distinct constraints on $n$ variables, and hence only finitely many instances. Therefore $\mathcal{S}_a$ and $\mathcal{S}_b$ are finite sets. Since they are finite, have the same cardinality, and satisfy $\mathcal{S}_a \subseteq \mathcal{S}_b$, it follows that $\mathcal{S}_a=\mathcal{S}_b$. Hence every instance satisfied by $b$ is also satisfied by $a$, i.e., $(b,a)\in \equivr$. Thus $\equivr$ is symmetric, and therefore an equivalence relation.

For the second claim, fix $a \in [q]^n$. For every $u \notin [a]_{\equivr}$, we have $(a,u)\notin \equivr$, so there exists an instance $J_u$ such that $a$ satisfies $J_u$ but $u$ does not. Let $P_a:=\bigcup_{u\notin [a]_{\equivr}} J_u$. Since the domain is finite, there are only finitely many such assignments $u$. If $c \in [a]_{\equivr}$, then $(a,c)\in \equivr$, and since $a$ satisfies every $J_u$, it follows that $c$ satisfies every $J_u$, and hence $c$ satisfies $P_a$. Conversely, if $u \notin [a]_{\equivr}$, then $u$ does not satisfy $J_u$, and therefore does not satisfy $P_a$. Thus $\Sat(P_a)=[a]_{\equivr}$.
\end{proof}

We now specify our reduction from $\Index$ for the main streaming lower bound. Intuitively, to prove our lower bound, starting from a maximum non-redundant instance, we take for each constraint $C_i$ a witness assignment $a_i$ that satisfies all other constraints but falsifies $C_i$. We use the above theorem to get the suffix $P_{a_i}$ whose set of satisfying assignments is exactly the assignments that are equivalent to $a_i$ under $\equivr$.
This is precisely what is needed for the reduction below.


\begin{proof}[Proof of \Cref{thm:main} \ref{main-b}]
Let $m := \nrdg$, and let $I = \{C_1,\dots,C_m\}$ be a non-redundant instance of $\CSP(\Gamma)$ on $n$ variables with exactly $m$ constraints. For each $i \in [m]$, fix a witness assignment $a_i \in [q]^n$ such that $a_i$ satisfies every constraint in $I \setminus \{C_i\}$ but does not satisfy $C_i$.
We now give a one-way communication protocol for $\Index_m$. Recall that Alice is given $x \in \{0,1\}^m$, Bob is given $i \in [m]$, and the goal is to recover $x_i$. Let $S_x := \{j \in [m] : x_j = 1\}$. Alice streams the sub-instance $I_{S_x} := \{C_j : j \in S_x\}$ through the streaming algorithm and sends Bob the resulting memory state.
Bob now appends the suffix $P_{a_i}$ from the \Cref{thm:equiv-lit}. If $x_i = 0$, then $i \notin S_x$, so every constraint of $I_{S_x}$ belongs to $I \setminus \{C_i\}$ and is therefore satisfied by $a_i$. Also, $a_i \in [a_i]_\equivr = \Sat(P_{a_i})$. Hence $I_{S_x} \cup P_{a_i}$ is satisfiable.
If $x_i = 1$, then $C_i \in I_{S_x}$. Moreover, every assignment in $\Sat(P_{a_i}) = [a_i]_\equivr$ satisfies exactly the same instances (and hence, constraints) as $a_i$. Since $a_i \unsats C_i$, it follows that no assignment in $[a_i]_\equivr$ satisfies $C_i$. Hence no assignment can satisfy $I_{S_x} \cup P_{a_i}$, so this instance is unsatisfiable.

Thus Bob can recover $x_i$ from the output of the streaming algorithm. Therefore any randomized single-pass streaming algorithm for satisfiability yields a public-coin one-way protocol for $\Index_m$ with the same message length, and thus by \Cref{lem:index-lb}, requires $\Omega(m)=\Omega(\NRD_n(\Gamma))$ space.
\end{proof}

\subsection{Positive Boolean \texorpdfstring{\csps}{CSPs}}\label{sec:lower-bool-nolit}

We now turn to the proof of \Cref{thm:main-bool-nolit} \ref{mainbool-b} that states the lower bound for positive Boolean \CSPs. Here the same relation $\equivr$ is still the right structural object, but the argument is more delicate: in the positive setting, symmetry need not hold automatically. The next theorem shows that, under the natural assumptions that $\Gamma$ is non constant-satisfiable, the relation $\equivr$ remains an equivalence relation once we restrict to non-constant assignments.


\begin{theorem}[Non-constant Equivalence Theorem for positive Boolean $\CSP$s]\label{thm:equiv-boolean-no-lit}
Let $n \in \mathbb{N}$ and $\Gamma$ be a non constant-satisfiable constraint language containing $k$-ary Boolean relations. For $a,b \in \{0,1\}^n$, define $(a,b)\in \equivr$ if every instance of the positive $\CSP(\Gamma)$ on $n$ variables that is satisfied by $a$ is also satisfied by $b$. Then the restriction of $\equivr$ to the non-constant assignments $\{0,1\}^n \setminus \{0^n,1^n\}$ is an equivalence relation.
\end{theorem}

\begin{proof}
Reflexivity and transitivity are immediate, so it suffices to prove symmetry. Fix distinct assignments $a,b \in \{0,1\}^n\backslash \{0^n,1^n\}$ with $(a,b)\in \equivr$. Define

\[E_0:=\{j \in [n]: a_j=b_j=0\}, \qquad E_1:=\{j \in [n]: a_j=b_j=1\}\]
\[ D_{01}:=\{j \in [n]: a_j=0,\ b_j=1\}, \qquad \ D_{10}:=\{j \in [n]: a_j=1,\ b_j=0\}.\]

Also define $P(a,b):=\{(a_j,b_j): j \in [n]\} \subseteq \{0,1\}^2$. By definition, $(a,b)\in \equivr$ iff for every relation $R \in \Gamma$ of arity $k$ the following holds:
\[\text{For every map }\psi:[k]\to P(a,b),
\text{ if } (\psi_1(1),\dots,\psi_1(k)) \in R, \text{  then } (\psi_2(1),\dots,\psi_2(k)) \in R,\] where $\psi_t(\ell)$ denotes the $t$-th coordinate of the pair $\psi(\ell)$. In words, whenever a tuple in $R$ is obtained by reading the first coordinates of some $k$ pairs from $P(a,b)$, the tuple obtained by reading the corresponding second coordinates must also lie in $R$.

Consider the following three cases:

\begin{enumerate}
    \item Suppose first that $D_{01}\neq \emptyset$ and $D_{10}=\emptyset$. Then $E_1\neq \emptyset$ because $a$ is non-constant, and $E_0\neq \emptyset$ because $b$ is non-constant and $D_{10}=\emptyset$. Fix any relation $R \in \Gamma$ and any tuple $u \in R$. We claim that $1^k \in R$. If $u \neq 1^k$, choose any coordinate $r$ with $u_r=0$, and define $\psi:[k]\to P(a,b)$ by setting $\psi(r):=(0,1)$, setting $\psi(\ell):=(1,1)$ for $\ell \neq r$ with $u_\ell=1$, and setting $\psi(\ell):=(0,0)$ for $\ell \neq r$ with $u_\ell=0$. Then $(\psi_1(1),\dots,\psi_1(k))=u \in R$, so $(\psi_2(1),\dots,\psi_2(k))=(u_1,\dots,u_{r-1},1,u_{r+1},\dots,u_k) \in R$. Repeating this argument coordinate by coordinate gives $1^k \in R$. Since $R$ was arbitrary, every relation in $\Gamma$ contains the all-$1$ tuple, contradicting that $\Gamma$ is non constant-satisfiable.
    \item  Suppose next that $D_{01}=\emptyset$ and $D_{10}\neq \emptyset$. An analogous argument shows that every relation in $\Gamma$ contains the all-$0$ tuple, again a contradiction.
    \item Therefore the only remaining possibility is that $D_{01}\neq \emptyset$ and $D_{10}\neq \emptyset$. In this case $P(a,b)=P(b,a)$, since both off-diagonal pairs $(0,1)$ and $(1,0)$ are present and the diagonal pairs are fixed under swapping. Hence the defining implication for $(a,b)\in \equivr$ is exactly the defining implication for $(b,a)\in \equivr$, and therefore $(b,a)\in \equivr$.
\end{enumerate}

Thus $\equivr$ is symmetric on $\{0,1\}^n \setminus \{0^n,1^n\}$, and its restriction to the non-constant assignments is an equivalence relation.
\end{proof}

Although we do not need it for the lower bound, the same argument also gives a simple description of the $\equivr$-equivalence classes on non-constant assignments.

\begin{remark}
For every non-constant assignment $a$, the $\equivr$-equivalence class of $a$ inside $\{0,1\}^n \setminus \{0^n,1^n\}$  is either $[a]_\equivr = \{a\}$ or $[a]_\equivr = \{a,\overline{a}\}$. 
\end{remark}

\begin{proof}
Fix distinct non-constant assignments $a,b$ with $(a,b)\in \equivr$. By the previous argument, both $D_{01}$ and $D_{10}$ must be non-empty. We claim that $E_0=E_1=\emptyset$. Indeed, if $E_0\neq\emptyset$, then using a coordinate from $E_0$ and one from $D_{10}$ using the same mapping argument as before, we conclude that every non-empty relation in $\Gamma$ must contain $0^k$, contradicting that $\Gamma$ is non constant-satisfiable. Similarly, if $E_1\neq\emptyset$, then using a coordinate from $E_1$ and one from $D_{01}$ forces every non-empty relation to contain $1^k$, again a contradiction. Hence every coordinate lies in either $D_{01}$ or $D_{10}$, so $b_j=1-a_j$ for all $j\in[n]$. Therefore $b=\overline{a}$.
\end{proof}

Unlike the general case, this restricted equivalence theorem is not by itself enough for the lower bound, because both the constant assignments must still be excluded explicitly. The next lemma strengthens the theorem in the form needed for the reduction. Given a non-constant witness assignment $a$ falsifying a constraint $C$, it constructs a suffix $P_{a,C}$ that is satisfied by $a$ and forces every satisfying assignment to falsify $C$.

\begin{lemma}\label{lem:constraint-suffix-boolean}
Let $C$ be a constraint in a positive $\CSP(\Gamma)$ instance, and let $a \in \{0,1\}^n \setminus \{0^n,1^n\}$ be such that $a \unsats C$. Then there exists an instance $P_{a,C}$ of positive $\CSP(\Gamma)$ such that $a \sats P_{a,C}$ and every assignment satisfying $P_{a,C}$ falsifies $C$.
\end{lemma}

\begin{proof}
 Since $\Gamma$ is non constant-satisfiable, there exist (not necessarily distinct) relations 
 $R^{(0)},R^{(1)} \in \Gamma$ such that $0^{k} \notin R^{(0)}$ 
 and $1^{k} \notin R^{(1)}$. Fix tuples $u \in R^{(0)}$ and $v \in R^{(1)}$. Since $a$ is non-constant, choose coordinates $p,q \in [n]$ 
 with $a_p=0$ and $a_q=1$. Let $H_a$ be the instance of positive \cspg on the two variables $x_p,x_q$ consisting of the following two constraints:
\begin{enumerate}
    \item the constraint $\{R^{(0)},(x'_1,\ldots,x'_k)\}$, where $x'_t=x_p$ if $u_t=0$ and $x'_t=x_q$ if $u_t=1$;
    \item the constraint $\{R^{(1)},(x''_1,\ldots,x''_k)\}$, where $x''_t=x_p$ if $v_t=0$ and $x''_t=x_q$ if $v_t=1$.
\end{enumerate}
Then $a \sats H_a$, while neither $0^n$ nor $1^n$ satisfies $H_a$.

Let $[a]_{\equivr}$ denote the $\equivr$-equivalence class of $a$ within the set $\{0,1\}^n\setminus\{0^n, 1^n\}$. For every non-constant assignment $b \notin [a]_{\equivr}$ such that $b \sats C$, Theorem~\ref{thm:equiv-boolean-no-lit} implies $(a,b)\notin \equivr$, so there exists an instance $J_b$ such that $a \sats J_b$ but $b \unsats J_b$. Define 
\[P_{a,C}:=H_a \cup \bigcup_{\substack{b \in \{0,1\}^n \setminus \{0^n,1^n\} \\ b \notin [a]_{\equivr},\ b \sats C}} J_b.\] Then $a \sats P_{a,C}$.

Now let $c$ be any assignment satisfying $P_{a,C}$. Suppose for contradiction that $c \sats C$. Since $c \sats H_a$, the assignment $c$ is non-constant. If $c \notin [a]_{\equivr}$, then $J_c$ appears in the union defining $P_{a,C}$, and $c \unsats J_c$, which is a contradiction. Hence $c \in [a]_{\equivr}$, so $(c,a)\in \equivr$. Since $c \sats C$, the one-constraint instance $\{C\}$ is satisfied by $c$, and therefore also by $a$, contradicting that $a \unsats C$. Thus every assignment satisfying $P_{a,C}$ \unsats $C$.
\end{proof}


We now complete the reduction. The point of \Cref{lem:constraint-suffix-boolean} is that whenever a constraint $C_i$ has a non-constant witness $a_i$, it yields a suffix $P_i$ that is satisfied by $a_i$ and forces every satisfying assignment to falsify $C_i$. Thus the same reduction as in the general case goes through, after separately handling the exceptional constraints that admit only constant witnesses.

\begin{proof}[Proof of \Cref{thm:main-bool-nolit} \ref{mainbool-b}]

Let $m:=\nrdg$, and let $I=\{C_1,\dots,C_m\}$ be a non-redundant instance of $\CSP(\Gamma)$ on $n$ variables with exactly $m$ constraints. Call a constraint $C_i$ bad if every witness for $C_i$ is a constant assignment (i.e., $0^n$ or $1^n$), and good otherwise. We claim that there are at most two bad constraints. Say three constraints were bad, then two of them would share the same constant witness assignment, which is impossible by definition. Let $G \subseteq [m]$ denote the set of good constraints. Then $|G| \ge m-2$. Using the fact that $m=\nrdg\geq \Omega(n)$, we get $|G|=\Omega(m)=\Omega(\nrdg)$.

Fix, for each $i \in G$, a non-constant witness $a_i \in \{0,1\}^n$ such that $a_i$ satisfies every constraint in $I \setminus \{C_i\}$ and falsifies $C_i$. By Lemma~\ref{lem:constraint-suffix-boolean}, for each $i \in G$ there exists an instance $P_i:=P_{a_i,C_i}$ such that $a_i \sats P_i$ and every assignment satisfying $P_i$ falsifies $C_i$.

We now give a one-way communication protocol for $\Index_{|G|}$. Enumerate $G=\{g_1,\dots,g_t\}$, where $t:=|G|$. Alice is given $x \in \{0,1\}^t$, Bob is given $r \in [t]$, and the goal is to recover $x_r$. Let $S_x:=\{g_j : x_j=1\} \subseteq G$. Recall that Alice streams the sub-instance $I_{S_x}:=\{C_j : j \in S_x\}$ through the streaming algorithm and sends Bob the resulting memory state. Bob sets $i:=g_r$ and appends the suffix $P_i$.

If $x_r=0$, then $i \notin S_x$, so every constraint of $I_{S_x}$ belongs to $I \setminus \{C_i\}$ and is therefore satisfied by $a_i$. Since $a_i \sats P_i$ as well, the instance $I_{S_x} \cup P_i$ is satisfiable.

If $x_r=1$, then $i \in S_x$, so $C_i \in I_{S_x}$. But every assignment satisfying $P_i$ falsifies $C_i$, by Lemma~\ref{lem:constraint-suffix-boolean}. Hence $I_{S_x} \cup P_i$ is unsatisfiable.

Thus Bob can recover $x_r$ from the output of the streaming algorithm. Therefore any randomized single-pass streaming algorithm for satisfiability yields a public-coin one-way communication protocol for $\Index_t$ with the same message length, and therefore by \Cref{lem:index-lb}, requires at least $\Omega(t)=\Omega(|G|)=\Omega(\nrdg)$ space.
\end{proof}

\subsection{Counterexamples over larger domains}\label{sec:counter-gen-nolit}

We now show that the lower bound for positive Boolean \csps does not extend to positive \csps over larger domains. This is not merely a failure of our proof technique via the equivalence relation $\equivr$; the lower bound itself is false in general. We give two counterexamples over non-Boolean domains. In each case, single-pass streaming satisfiability is strictly easier than what $\nrdg$ would suggest, and the relation $\equivr$ also fails to be symmetric.

\subsubsection{Non-Singleton constraint language}\label{sec:counter-anyarity}

Let the domain be $[3]=\{0,1,2\}$, and let the language be:
\[\Gamma := \{R,U_1,U_2\}, \quad \text{ such that }\quad  R := [3]^2 \setminus \{(0,0)\}, \quad U_1 := \{1\}, \quad U_2 := \{2\} \]

Clearly $\Gamma$ is non constant-satisfiable (the all-$0$ assignment violates both $U_1$ and $U_2$, the all-$1$ assignment violates $U_2$, and the all-$2$ assignment violates $U_1$).

\paragraph{Linear Streaming Satisfiability.} 
We first observe that streaming satisfiability for positive $\CSP(\Gamma)$ is easy. Observe that once a variable is assigned a nonzero value, every binary constraint of type $R$ involving that variable is automatically satisfied, since every pair in $\{1,2\}^2$ belongs to $R$. Therefore an instance is satisfiable if and only if no variable receives both unary constraints $\{U_1,(x_i)\}$ and $\{U_2,(x_i)\}$. This gives a single-pass $O(n)$-space algorithm: for each variable, store whether it has received a $U_1$-constraint and whether it has received a $U_2$-constraint, and reject if both occur. If no conflict occurs, then assigning each variable a value in $\{1,2\}$ consistent with its unary constraints satisfies every constraint in the instance.

\paragraph{Quadratic \NRD.} 
On the other hand, $\nrdg$ is quadratic. Consider the instance
$I := \{\{R,(x_i,x_j)\} : 1 \le i < j \le n\}$.
We claim that $I$ is non-redundant. Fix any constraint $\{R,(x_i,x_j)\}$ in $I$. Set $x_i=x_j=0$ and set every other variable to $1$. This assignment falsifies the chosen constraint, since $(0,0)\notin R$, but satisfies every other constraint in $I$, because every other pair has at least one endpoint equal to $1$. Thus every constraint in $I$ has a witness assignment, and so $I$ is non-redundant. It follows that $\nrdg\ge \binom{n}{2}=\Omega(n^2)$. Since each relation in $\Gamma$ has arity at most $2$, there are only $O(n^2)$ distinct binary constraints and $O(n)$ unary constraints on $n$ variables, so trivially $\nrdg=O(n^2)$. Hence $\nrdg=\Theta(n^2)$.

\paragraph{No Equivalence Relation.} 
We next observe that the relation $\equivr$ is not symmetric for this language. Let $a=(0,1)\in [3]^2$ and $b=(2,1)\in [3]^2$. Every atomic constraint satisfied by $a$ is also satisfied by $b$. In particular, $\{U_1,(x_2)\}$ is satisfied by both assignments, and every binary $R$-constraint satisfied by $a$ remains satisfied after changing the first coordinate from $0$ to $2$, since the only forbidden tuple of $R$ is $(0,0)$. It follows that $(a,b)\in \equivr$. However, $b$ satisfies the unary constraint $\{U_2,(x_1)\}$, while $a$ does not. Therefore $(b,a)\notin \equivr$, and hence $\equivr$ is not symmetric.

We conclude that for the positive \cspg defined by the non constant-satisfiable language $\Gamma$ as above, single-pass streaming satisfiability has space complexity $O(n)$ while $\nrdg=\Theta(n^2)$. In particular, the lower bound by non-redundancy fails in general over larger domains, and the equivalence-theorem mechanism also breaks down.

The same construction extends immediately to every fixed arity $r\ge 2$. Indeed, if one replaces $R$ by $R_r := [3]^r \setminus \{0^r\}$ and keeps the same unary relations $U_1$ and $U_2$, then the same proof gives a positive \csp with single-pass streaming satisfiability complexity $O(n)$ and non-redundancy $\Theta(n^r)$.

\subsubsection{Singleton constraint language}\label{sec:counter-singleton}

We now show that the same phenomenon occurs even for a language consisting of a single binary relation. Let the domain be $[4]=\{0,1,2,3\}$, and let $\Gamma:=\{T\}$, where $T\subseteq [4]^2$ is defined by
$T(x,y)$ iff $(x\in\{2,3\}\ \text{or}\ y\in\{2,3\})$ and $x \not\equiv y \pmod 2$.
Equivalently,
\[\Gamma:=\{T\}, \quad \text{ such that } \quad T=\{(0,3),(3,0),(1,2),(2,1),(2,3),(3,2)\}.\]
Clearly $\Gamma$ is non constant-satisfiable, since $(c,c)\notin T$ for every $c\in[4]$.

\paragraph{Linear Streaming Satisfiability.} 
We claim that satisfiability of the positive $\CSP(\Gamma)$ is exactly graph bipartiteness. Given an instance, form its underlying graph on vertex set $[n]$, where each constraint $\{T,(x_i,x_j)\}$ is viewed as an edge between $i$ and $j$. If the instance is satisfiable, then every edge $\{T,(x_i,x_j)\}$ forces the assigned values on $x_i$ and $x_j$ to have opposite parity, so the parity of the assignment gives a proper $2$-coloring of the graph. Conversely, if the graph is bipartite, then assigning the value $2$ to one side and the value $3$ to the other satisfies every edge, since $(2,3),(3,2)\in T$. Therefore satisfiability of $\CSP(\Gamma)$ is equivalent to graph bipartiteness. Since single-pass streaming bipartiteness can be decided in $O(n\log n)$ space, the same upper bound holds for streaming satisfiability of $\CSP(\Gamma)$.

\paragraph{Quadratic \NRD.}
We now show that $\nrdg=\Theta(n^2)$. Let $L\cup R=[n]$ be a bipartition with $|L|=\lfloor n/2\rfloor$ and $|R|=\lceil n/2\rceil$, and consider the instance
$I := \{\{T,(x_i,x_j)\} : i\in L,\ j\in R\}$,
corresponding to the complete bipartite graph between $L$ and $R$. We claim that $I$ is non-redundant. Fix any constraint $\{T,(x_i,x_j)\}$ with $i\in L$ and $j\in R$. Assign the value $0$ to $x_i$, the value $1$ to $x_j$, the value $2$ to every other variable in $L$, and the value $3$ to every other variable in $R$. Then the chosen constraint is falsified, since $(0,1)\notin T$, while every other constraint in $I$ is satisfied: every remaining edge has endpoints of opposite parity, and at least one endpoint lies in $\{2,3\}$. Thus every constraint in $I$ has a witness assignment, so $I$ is non-redundant. Hence $\nrdg\ge |L||R|=\Omega(n^2)$. Since each relation in $\Gamma$ has arity at most $2$, trivially $\nrdg=O(n^2)$, and therefore $\nrdg=\Theta(n^2)$.

\paragraph{No Equivalence Relation.}
Finally, we observe that $\equivr$ is not symmetric here either. Let $a=(0,1,3)\in [4]^3$ and $b=(2,1,3)\in [4]^3$. Under $a$, the only ordered pairs of coordinates that satisfy $T$ are $(1,3)$ and $(3,1)$, since $(a_1,a_3)=(0,3)\in T$ and $(a_3,a_1)=(3,0)\in T$. Under $b$, these two ordered pairs are still satisfied, and in addition $(b_1,b_2)=(2,1)\in T$ and $(b_2,b_1)=(1,2)\in T$. Thus every atomic constraint satisfied by $a$ is also satisfied by $b$, so $(a,b)\in \equivr$. However, $b$ satisfies the constraint $\{T,(x_1,x_2)\}$, while $a$ does not. Therefore $(b,a)\notin \equivr$, and hence $\equivr$ is not symmetric.

We conclude that even for the positive \cspg defined by the above singleton, non constant-satisfiable binary relation $T$ over a domain of size $4$, the lower bound by non-redundancy fails, and the relation $\equivr$ need not be an equivalence relation.

\section{Conclusion}\label{sec:conclusion}

We study the single-pass streaming satisfiability of \CSPs through the lens of 
non-redundancy. Our main result shows that the streaming satisfiability of finite \CSPs is tightly 
characterized by non-redundancy up to a logarithmic factor, with an analogous characterization 
holding for positive Boolean \CSPs that are non constant-satisfiable. On the 
algorithmic side, we give a deterministic upper bound by maintaining an equivalent 
non-redundant sub-instance throughout the stream. On the lower bound side, we introduce 
a relation on assignments, show that it is an equivalence relation in the 
settings above, and use this structure to obtain tight reductions from a communication
problem.

There are several natural directions for future work. First, it would be interesting 
to understand \emph{streaming approximate sparsification} through the same lens. In 
the classical offline setting, recent work of Brakensiek and 
Guruswami~\cite{brakensiek2025redundancy} shows that non-redundancy governs approximate 
sparsification up to polylogarithmic factors for unweighted \CSPs, though their result 
is existential and non-constructive. A natural open question is therefore whether 
streaming algorithms can maintain small approximate sparsifiers in $\widetilde{O}(\nrdg)$ 
space, or whether fundamental barriers emerge in the streaming setting.

Second, our counterexamples over larger domains show that positive \CSPs cannot be characterized by non-redundancy in full generality. Even after excluding the trivial cases of empty relations and constant-satisfiable languages, there are positive constraint languages for which single-pass streaming satisfiability is strictly easier than what $\NRD_n(\Gamma)$ predicts.
This suggests that the characterization for positive \csps on non-Boolean domains could be quite different from our characterization, and is an interesting open problem.  


Third, our results further highlight the importance of obtaining tight bounds on 
$\nrdg$ itself. Beyond the streaming setting, non-redundancy has emerged as a central 
structural parameter governing sparsification and kernelization. A better understanding 
of the possible asymptotic behaviours of $\nrdg$, together with more explicit 
classifications for natural predicate families, would directly sharpen the consequences 
of our streaming characterization.\label{line:last}

\section*{Acknowledgement}
We thank Joshua Brakensiek, Aaron Putterman, and Venkatesan Guruswami for helping us better
understand the context of \nrd in the literature.

\printbibliography

\end{document}

%% file: probheader.tex
\input{header}

\usepackage{silence}
\usepackage[colorinlistoftodos,textsize=tiny,textwidth=2cm,color=green!50!gray]{todonotes}

\def\DEBUG{true}

\ifdefined\DEBUG

\newcommand{\amati}[1]{\textcolor{brown}{[AS: #1]}}
\newcommand{\amat}[1]{\textcolor{brown}{#1}}
\newcommand{\amatr}[1]{\todo[color=brown!100!black!20]{AS: #1}}

\newcommand{\snote}[1]{\textcolor{red}{[SV: #1]}}
\newcommand{\santr}[1]{\todo[color=red!100!black!20]{SV: #1}}
\newcommand{\todosec}[1]{\textcolor{orange}{#1}}
\else
\newcommand{\amat}[1]{#1}
\newcommand{\amati}[1]{}
\newcommand{\amatr}[1]{}

\newcommand{\snote}[1]{}
\newcommand{\santi}[1]{}
\newcommant{\todosec}[1]{}
\fi

\usepackage[style=alphabetic, backend=biber, minalphanames=3, maxalphanames=4, maxbibnames=99, maxcitenames=4]{biblatex}

\input{thmmacros}

\newcommand{\CSP}{\ensuremath{\mathsf{CSP}}\xspace}
\newcommand{\csp}{\CSP}

\newcommand{\maxcsp}{\ensuremath{\mathsf{Max}\text{-}\mathsf{CSP}}\xspace}
\newcommand{\mincsp}{\ensuremath{\mathsf{Min}\text{-}\mathsf{CSP}}\xspace}

\newcommand{\CSPs}{\ensuremath{\mathsf{CSPs}}\xspace}
\newcommand{\csps}{\CSPs}

\newcommand{\maxcsps}{\ensuremath{\mathsf{Max}\text{-}\mathsf{CSP}}\text{s}\xspace}
\newcommand{\mincsps}{\ensuremath{\mathsf{Min}\text{-}\mathsf{CSP}}\text{s}\xspace}

\newcommand{\NRD}{\ensuremath{\mathsf{NRD}}\xspace}
\newcommand{\nrd}{\NRD}

\newcommand{\poly}{\ensuremath{\operatorname{poly}}}

\newcommand{\cspg}{\ensuremath{\CSP(\Gamma)}\xspace}

\newcommand{\nrdx}[1]{\ensuremath{\NRD_n(#1)}\xspace}

\newcommand{\nrdg}{\nrdx{\Gamma}}

\newcommand{\Sat}{\mathrm{Sat}} 
\newcommand{\sats}{\ensuremath{\ \mathrm{satisfies}\ }}
\newcommand{\unsats}{\ensuremath{\ \mathrm{unsatisfies}\ }}

\newcommand{\twosat}{$2$\textsf{-SAT}}
\newcommand{\twolin}{$2$\textsf{-LIN}}

\newcommand{\ksat}{$k$\textsf{-SAT}}

\newcommand{\lin}{\textsf{LIN}\xspace}

\newcommand{\maxx}{\textsf{Max}}

\newcommand{\minn}{\textsf{Min}}

\newcommand{\equivr}{\ensuremath{\mathcal{E}}}

\newcommand{\Index}{\ensuremath{\mathsf{INDEX}}\xspace}



\allowdisplaybreaks

    \title{Characterizing Streaming Decidability of CSPs via Non-Redundancy}
\author{Anonymous authors}

\title{Characterizing Streaming Decidability of CSPs via Non-Redundancy}
\author{Amatya Sharma\thanks{University of Michigan, Ann Arbor, Michigan, USA. Email: \texttt{amatya@umich.edu}. Work done in part when the author was visiting the Toyota Technological Institute at Chicago as an intern in Fall 2025.}\and{Santhoshini Velusamy \thanks{University of Waterloo, Ontario, Canada. Supported by NSERC Discovery grant RGPIN-2026-04689 and ECR supplement DGECR-2026-00421. Supported in part by NSF award CCF 2348475 when the author was affiliated with Toyota Technological Institute at Chicago. Email: \texttt{santhoshini.velusamy@uwaterloo.ca}}}}

\date{}

%% file: header.tex
\usepackage[margin=1in]{geometry}
\usepackage[utf8]{inputenc}
\usepackage{amsthm}
\usepackage{amsthm, amsmath, amssymb, mathtools}
\usepackage{bm,bbm}
\usepackage{graphicx}
\usepackage{color,xcolor}
\usepackage{paralist,enumitem}
\usepackage[normalem]{ulem}
\usepackage{tabularx}
\usepackage{tikz}
\usepackage{caption,comment}
\usepackage{algorithmicx}
\usepackage{algorithm}
\usepackage{algpseudocode}
\usepackage{lipsum}
\usepackage[colorlinks=true, allcolors=blue]{hyperref}
\usepackage[capitalize,noabbrev,nameinlink]{cleveref}
\usepackage{booktabs}
\usepackage{subcaption}

\usepackage{xspace}

\usepackage{lineno}
\nolinenumbers

\usepackage{mathpazo}
\usepackage{microtype}
\newcounter{algsubstate}
\renewcommand{\thealgsubstate}{\alph{algsubstate}}

\algnewcommand\algorithmicinput{\textbf{Input:}}
\algnewcommand\Input{\item[\algorithmicinput]}

\algnewcommand\algorithmicoutput{\textbf{Output:}}
\algnewcommand\Output{\item[\algorithmicoutput]}

\algnewcommand\algorithmicgoal{\textbf{Goal:}}
\algnewcommand\Goal{\item[\algorithmicgoal]}

\newcommand{\blind}[2]{{\ifnum\draft=1\color{purple}\fi \ifnum\doubleblind=1#2\fi\ifnum\doubleblind=0#1\fi\ifnum\doubleblind=2$\{$ #1 $\vert$ #2 $\}$\fi}}

\makeatletter
\newcommand{\algmargin}{\the\ALG@thistlm}
\algnewcommand{\parState}[1]{\State%
  \parbox[t]{\dimexpr\linewidth-\algmargin}{\strut #1\strut}}

%% file: thmmacros.tex
\usepackage{amsthm}
\usepackage{thmtools,thm-restate}

\numberwithin{equation}{section}
\declaretheoremstyle[bodyfont=\it,qed=\qedsymbol]{noproofstyle}

\declaretheorem[name=Observation,numbered=no]{observation*}

\declaretheorem[numberlike=equation]{theorem}

\declaretheorem[name=Theorem,numbered=no]{theorem*}

\declaretheorem[numberlike=equation]{lemma}
\declaretheorem[name=Lemma,numbered=no]{lemma*}

\declaretheorem[name=Corollary,numbered=no]{corollary*}

\declaretheorem[name=Proposition,numbered=no]{proposition*}

\declaretheorem[name=Claim,numbered=no]{claim*}

\declaretheorem[name=Conjecture,numbered=no]{conjecture*}

\declaretheorem[name=Question,numbered=no]{question*}

\declaretheoremstyle[spaceabove=12pt,spacebelow=12pt,bodyfont=\it]
{defstyle}

\declaretheorem[unnumbered,name=Definition,style=defstyle]{definition*}

\declaretheorem[unnumbered,name=Example,style=defstyle]{example*}

\declaretheorem[unnumbered,name=Notation=defstyle]{notation*}

\declaretheorem[unnumbered,name=Construction,style=defstyle]{construction*}

\declaretheoremstyle[]{rmkstyle} 

\newtheorem*{remark}{Remark}

\crefname{claim}{claim}{claims}
\crefname{fact}{fact}{facts}
\crefname{corollary}{corollary}{corollaries}
\crefname{theorem}{theorem}{theorems}
\crefname{table}{table}{tables}

\Crefname{theorem}{Theorem}{Theorems}
\Crefname{claim}{Claim}{Claims}
\Crefname{fact}{Fact}{Facts}
\Crefname{corollary}{Corollary}{Corollaries}
\Crefname{table}{Table}{Tables}
